\documentclass[USenglish]{article}

\usepackage[utf8]{inputenc}

\usepackage[small]{dgruyter}
\usepackage{microtype}

\usepackage{hyperref}

\newtheorem{theorem}{Theorem}[section]
\newtheorem{lemma}[theorem]{Lemma}

\newtheorem{proposition}[theorem]{Proposition}

\theoremstyle{definition}
\newtheorem{definition}[theorem]{Definition}
\newtheorem{example}[theorem]{Example}

\theoremstyle{remark}

\numberwithin{equation}{section}

\begin{document}


  \author[1]{Alessandro De Gregorio}
  \affil[1]{Department of Mathematics, University of Bologna, e-mail: alessandro.de@studio.unibo.it}
  \title{On the set of optimal homeomorphisms for the natural pseudo-distance associated with the Lie group $S^1$.}
  \runningtitle{On the set of optimal homeomorphisms.}
  \abstract{If $\varphi$ and $\psi$ are two continuous real-valued functions defined on a compact topological space $X$ and $G$ is a subgroup of the group of all homeomorphisms of $X$ onto itself, the natural pseudo-distance $d_G(\varphi,\psi)$ is defined as the infimum of $\mathcal{L}(g)=\|\varphi-\psi \circ g \|_\infty$, as $g$ varies in $G$. 
  In this paper, we make a first step towards extending the study of this concept to the case of Lie groups, by assuming $X=G=S^1$. 
  In particular, we study the set of the optimal homeomorphisms for $d_G$, i.e. the elements $\rho_\alpha$ of $S^1$ such that $\mathcal{L}(\rho_\alpha)$ is equal to $d_G(\varphi,\psi)$. As our main results, we give conditions that a homeomorphism has to meet in order to be optimal, and we prove that the set of the optimal homeomorphisms is finite under suitable conditions.}
  \startpage{1}

\maketitle
\section{Introduction} 
In the last 25 years persistent homology has spread across the mathematical community as an algebraic topological method for the study of spaces of real functions and their topological properties. Persistent homology is an extension of classical homology,
revisited from a point of view inspired by Morse theory and the concept of ``persistence'' of homological features in presence of noise. It has proved itself a valuable tool also for the application to other
sciences, as a way to analyze data or shapes, due to its robustness to noise and other artifacts. The reader can find a presentation of this theory in \cite{edel0}.  

After fixing a continuous real-valued function $\varphi$, from now on called \emph{filtering function}, we can consider the sublevel sets of $\varphi$ and their homology groups. By taking track of birth and death of the homology classes, due to the 
topological changes of the sublevel sets, it is constructed a multiset called \emph{persistence diagram} that outlines some of the topological properties of the considered function. Details related to this process are presented in the survey \cite{edel1}.  

One of the tools available to study persistent homology is  the concept of \emph{natural pseudo-distance}, a dissimilarity measure between filtering functions. Given a compact topological space $X$, and two continuous functions $\varphi,\psi:X \rightarrow \mathbb{R}$ together with a subgroup $G$ of the group $H$ of all homeomorphisms from $X$ onto $X$, the natural pseudo-distance associated with $G$ between $\varphi$ and $\psi$ is defined as 
\begin{equation}
\label{npd}
d_{G}(\varphi, \psi)= \inf_{g\in G}\max_{P \in X}|\varphi(P) - \psi (g(P)) |.
\end{equation}
The paper \cite{prof4} studies the natural pseudo-distance in the case that $G=H$ and $X$ is a closed manifold. If $G=H$ and $X$ is a closed curve, the natural pseudo-distance is strictly related to the Fr\'echet distance \cite{alt}. We refer the reader to \cite{prof3} for more details concerning the case of surfaces.

The main link between the natural pseudo-distance and persistent homology is a \emph{stability theorem}, which guarantees that the natural pseudo-distance is an upper-bound for the \emph{matching distance} between persistence diagrams. A study of this relation can be found in \cite{prof5}.

In this article we will confine ourselves to study functions defined on the Lie group $S^1$, when the group of homeomorphisms acting on these functions is the group of rotation $S^1$ itself. This is justified by several reasons, one of them being the interest for the study of the persistent homology of curves.  
For example, the length theorem for curves, a generalization of the Fary's theorem, can be proved by means of the stability theorem in persistent homology \cite{edel2}. We underline that the problem of analyzing 
dissimilarities between closed curves has been object of study from the point of view of persistent homology (cf., e.g., \cite{carlsson1} and \cite{prof1}) not only for the sake of the 
interest in these themes, but also to enlighten some of the aspects relevant to the generalization to higher dimensional cases. Several results
have been developed to estimate the natural pseudo-distance in the case the considered space is a closed curve (cf., e.g., \cite{difabio1} and 
\cite{prof2}). Closed curves attract interest also from an applicative point of view, since they model in a natural way the problem of analyzing waveform signals and periodic functions in one variable \cite{hahn1}. In this kind of problems the choice of the group of homeomorphisms plays a fundamental role, and some proper subgroup of the group of all homeomorphisms, as in this case the group of rotation $S^1$, is more suitable to represent significant aspects of these matters. 
Last but not least, the study of the natural pseudo-distance associated with the Lie group $S^1$ aims to be a first step towards the extension of persistent homology to Lie group theory.

An aspect related to the natural pseudo-distance is that of \emph{optimality}. With the need to estimate the dissimilarities between filtering 
functions comes also the need to find the homeomorphisms corresponding to the best matchings between the two functions, called \emph{optimal homeomorphisms}. 
The concept of optimality accomplishes a relevant role from both a theoretical and an applicative point of view, enabling us to understand whether or not there
may exist a point-to-point correspondence between two sets that is more significant than any other correspondence. It is therefore important to 
identify the cases where the number of such homeomorphisms is finite. In these conditions the considered correspondences reflect determinant 
properties of the functions under analysis and, from a computational point of view, it is of great use to have only a finite number of cases to consider.
The concept of optimality is not new in literature concerning persistence homology. In \cite{prof4} it has been shown that the existence of optimal homeomorphisms is not always granted. 
The research carried out in \cite{fasy1} improves the results in \cite{edel2}, paying attention to the existence of the homeomorphisms that realize the Fr\'echet distance between two curves.
In \cite{cerri1} the authors investigate the problem of optimality when the considered space is a closed curve, showing that if the two filtering functions are Morse and the natural pseudo-distance between them is $0$, then there exists at least one optimal diffeomorphism.

The main contribution of this paper concerns the study of the properties of optimal homeomorphisms, the necessary conditions that they have to satisfy, and the proof that the set of optimal homeomorphisms is finite, under the assumption that the filtering functions on $S^1$ are Morse and the invariance group is $S^1$ itself.\\
    
The outline of the paper is as follows. In Section \ref{sec2} we recall the main definitions and analyze the properties that have to be met by a homeomorphism in order to be optimal. In Section \ref{sec3} we prove the main theorems of our paper, which state some differential conditions that are necessary for optimality, and the finiteness of the set of optimal homeomorphisms, provided that  the filtering functions are Morse.

\section{Natural pseudo-distance and optimality}
\label{sec2}
Let us consider the Lie group $S^1$. In order to simplify our calculations we will see the manifold $S^1$ as the quotient space 
$\mathbb{R}/\!\!\!\sim$, were $x \sim y$ if and only if $x-y$ is a multiple of $2\pi$.
Every element $\alpha \in S^1$ can also be identified with the homeomorphism $\rho_\alpha: S^1 \rightarrow S^1$ defined by setting 
$\rho_\alpha(\theta)=\theta+\alpha$.  We are interested in studying the space $\mathcal{M}$ of all Morse functions from $S^1$ to $\mathbb{R}$.
If in (\ref{npd}) we set $G=S^1$ we get the following \emph{natural pseudo-distance} on $\mathcal{M}$:

\begin{equation}
d_{S^1}(\varphi,\psi) = \inf_{\rho_\alpha \in S^1}|| \varphi - \psi \circ \rho_\alpha||_\infty .
\end{equation}

Our paper focuses on the concept expressed by the following definition.
\begin{definition}
We say that $\rho_{\bar{\alpha}}$ is an \emph{optimal} homeomorphism between $\varphi$ and $\psi$ if 
\begin{equation}
 d_{S^1}(\varphi,\psi) = \max_{\theta \in S^1}|\varphi(\theta) - \psi(\rho_{\bar{\alpha}}(\theta)) | .
\end{equation}
\end{definition}

Given that $S^1$ is a compact group with respect to the usual topology, it is trivial to observe that at least one optimal homeomorphism between $\varphi$ and $\psi$
always exists in $S^1$.

For the sake of simplicity we define the following functions that we will often use in this section:
\begin{definition}
Given $\varphi , \psi \in \mathcal{M}$ we will set
\begin{gather*}
F: S^1 \times S^1 \rightarrow \mathbb{R} \\
F(\theta,\alpha) = |\varphi(\theta)-\psi(\theta+\alpha)|
\end{gather*} and
\begin{gather*}
f_\alpha : S^1  \rightarrow \mathbb{R} \qquad \text{ for each fixed } \alpha \in S^1  \\
f_\alpha(\theta) = F(\theta,\alpha).
\end{gather*}
If $\rho_{\bar{\alpha}}$ is an optimal homeomorphism, we say that a  point $(\bar{\theta},\bar{\alpha}) \in S^1 \times S^1$ \emph{realizes the natural pseudo-distance} if 
\begin{equation*}
d_{S^1}(\varphi,\psi) = F(\bar{\theta},\bar{\alpha}).
\end{equation*}
\end{definition}

We observe that our functions $F$ and $f_\alpha$ (for all $\alpha$ in $S^1$) are uniformly continuous functions because of the compactness of their domain. 
Moreover, at the points $(\theta,\alpha)$ where $F(\theta,\alpha) \neq 0$, $F$ is at least a $C^1$ function and the same holds for 
$f_\alpha$ at the points $\theta$ where $f_\alpha (\theta) \neq 0$. In this section we will confine ourselves to consider the case in which 
the natural pseudo-distance is strictly greater than 0; in such a hypothesis the 
function $F$ is $C^1$ at the points that realize the natural pseudo-distance. The case in which $d_{S^1}(\varphi,\psi)=0$ will be examined in the theorems in the next section. To simplify our notation we will write $d_{S^1}(\varphi,\psi)=\bar{d}$.

\begin{lemma}
\label{lem1}
Let $\rho_{\bar{\alpha}}$ be an optimal homeomorphism and $\bar{\alpha}$ be the element of $S^1$ associated with it. Then there exist two points $\bar{\theta}_1$ and $\bar{\theta}_2$ of $S^1$ and two sequences $(\theta_i^u,\alpha_i^u)$, $(\theta_i^d,\alpha_i^d)$ in $S^1 \times S^1$ with the following properties:
\begin{gather*}
\alpha_i^u \xrightarrow{i \rightarrow \infty} \bar{\alpha} \qquad \alpha_i^u > \bar{\alpha} \\
\theta_i^u \xrightarrow{i \rightarrow \infty} \bar{\theta}_1 \text{ with } \\
f_{\alpha_i^u}(\theta_i^u) = \max_{\theta \in S^1}f_{\alpha_i^u}(\theta) \qquad
f_{\bar{\alpha}}(\bar{\theta}_1) = \max_{\theta \in S^1}f_{\bar{\alpha}}(\theta) 
\end{gather*}  
and
\begin{gather*}
\alpha_i^d \xrightarrow{i \rightarrow \infty} \bar{\alpha} \qquad \alpha_i^d < \bar{\alpha} \\
\theta_i^d \xrightarrow{i \rightarrow \infty} \bar{\theta}_2 \text{ with } \\
f_{\alpha_i^d}(\theta_i^d) = \max_{\theta \in S^1}f_{\alpha_i^d}(\theta) \qquad
f_{\bar{\alpha}}(\bar{\theta}_2) = \max_{\theta \in S^1}f_{\bar{\alpha}}(\theta) .
\end{gather*}  
\end{lemma}
\begin{proof}
We will prove the existence of the first sequence as for the other one the same arguments can be used. We can take a sequence 
$(\alpha_i^u)$ such that $\lim_{i \rightarrow \infty}\alpha_i^u=\bar{\alpha}$ and $\alpha_i^u > \bar{\alpha}$ for every index $i$. Since every function $f_\alpha$ is defined on $S^1$, 
and therefore has a compact domain, for each $i$ it is possible to choose a point $\theta_i^u$ of global maximum for the function $f_{\alpha_i^u}$. 
The resulting sequence $(\theta_i^u)$ is defined in a compact space, hence we can take a convergent subsequence $(\theta_{i_j}^u)$ with limit 
a point $\bar{\theta}_1$. Associating with each $(\theta_{i_j}^u)$ the corresponding $\alpha_{i_j}^u$ we obtain a sequence $(\theta_{i_j}^u,\alpha_{i_j}^u)$ 
that converges to $(\bar{\theta}_1,\bar{\alpha})$. It remains to prove that $\bar{\theta}_1$ is a point of global maximum for the function $f_{\bar{\alpha}}$. 
By contradiction suppose that $\bar{\theta}_1$ is not a point of global maximum. The function $f_{\bar{\alpha}}$ must have at least one point of absolute maximum, 
let us call it $\tilde{\theta}$. Then, since $F$ is a uniformly continuous function, there must exist a neighborhood $U$ of $(\tilde{\theta},\bar{\alpha})$ 
and a neighborhood $W$ of $(\bar{\theta}_1,\bar{\alpha})$ such that on $U$ the function $F$ takes values strictly greater than the ones taken on $W$. 
Then for some of the $f_{\alpha_{i_j}^u}$ we can find some points $\tilde{\theta}_{i_j}^u$ such that $f_{\alpha_{i_j}^u}(\tilde{\theta}_{i_j}^u)>f_{\alpha_{i_j}^u}(\theta_{i_j}^u)$ 
and that is absurd for how we chose the sequence $(\theta_{i_j}^u)$.
\end{proof}

We can now give conditions that the points $(\bar{\theta},\bar{\alpha})$ realizing the natural pseudo-distance have to satisfy.

\begin{theorem}
\label{p1}
Let us take an optimal homeomorphism $\rho_{\bar{\alpha}}$ and assume that the function $f_{\bar{\alpha}}$ has only one point of absolute maximum $\bar{\theta}$, i.e. only one point $\bar{\theta}$ in $S^1$ such that $f_{\bar{\alpha}}(\bar{\theta})=d_{S^1}(\varphi,\psi)$. Then $(\bar{\theta},\bar{\alpha})$ is a critical point for the function $F$:
\begin{equation}
\frac{\partial F}{\partial \theta}(\bar{\theta},\bar{\alpha})= 0 \qquad \text{and} \qquad \frac{\partial F}{\partial \alpha}(\bar{\theta},\bar{\alpha})= 0.
\end{equation}
\end{theorem}
\begin{proof}
At first we see that the function $F$ is  at least $C^1$ in a neighborhood of $(\bar{\theta},\bar{\alpha})$, because otherwise we would have $F(\bar{\theta},\bar{\alpha})=0$ and then $f_{\bar{\alpha}}$ would be the constant function equal to $0$. This is absurd since $f_{\bar{\alpha}}$ has only one point of absolute maximum.
We immediately see that it must be  $\frac{\partial F}{\partial \theta}(\bar{\theta},\bar{\alpha})= \frac{d}{d\theta}f_{\bar{\alpha}}(\bar{\theta})= 0$, since $\bar{\theta}$ is a point of global maximum for the function $f_{\bar{\alpha}}$. Now, let us consider the two sequences $(\theta_i^u,\alpha_i^u)$, $(\theta_i^d,\alpha_i^d)$ in $S^1 \times S^1$, whose existence is guaranteed by Lemma \ref{lem1}. Since $\bar{\theta}$ is the only point of absolute maximum for the function $f_{\bar{\alpha}}$ it must be
\begin{align*}
(\theta_i^u,\alpha_i^u)& \xrightarrow{i \rightarrow \infty} (\bar{\theta},\bar{\alpha}) \quad \text{and} \\
(\theta_i^d,\alpha_i^d)& \xrightarrow{i \rightarrow \infty} (\bar{\theta},\bar{\alpha}). 
\end{align*}
Therefore, since for all $i$ we have $F(\theta_i^u,\alpha_i^u) \geq F(\bar{\theta},\bar{\alpha}) \geq F(\theta_i^u,\bar{\alpha})$ and similarly $F(\theta_i^d,\alpha_i^d) \geq F(\bar{\theta},\bar{\alpha})\geq F(\theta_i^d,\bar{\alpha})$, we see that

\begin{align*}
\lim_{i \rightarrow \infty}\frac{F(\theta_i^u,\alpha_i^u) -F(\theta_i^u,\bar{\alpha})}{\alpha_i^u- \bar{\alpha}} \geq 0 \\
\lim_{i \rightarrow \infty}\frac{F(\theta_i^d,\alpha_i^d) -F(\theta_i^d,\bar{\alpha})}{\alpha_i^d- \bar{\alpha}} \leq 0 
\end{align*}
and since both the limits converge to $\frac{\partial F}{\partial \alpha}(\bar{\theta},\bar{\alpha})$ it must be $\frac{\partial F}{\partial \alpha}(\bar{\theta},\bar{\alpha})= 0$.
\end{proof}

\begin{example}
As an example illustrating the statement of Theorem \ref{p1} we can consider the two functions $\varphi(\theta)=\frac{1}{2}\sin^2(\frac{\theta}{2})$ and $\psi(\theta)=\sin^2(\frac{\theta}{2})$. In this case the natural pseudo-distance is $\frac{1}{2}$ and it is realized by the point $(0,\pi)$, which is a critical point for $F(\theta,\alpha)$.
\end{example}
We can also find examples where $\rho_\alpha$ is  an optimal homeomorphism and the points $(\theta, \alpha)$ that realize the natural pseudo-distance are not critical points for $F$.
\begin{example}
\label{e2}
Consider the two functions $\varphi$ and $\psi$ drawn in Figure \ref{fig1}. We see that the optimal homeomorphism is the identity and that the points that realize
the natural pseudo-distance are $(\frac{\pi}{2},0)$ and $(\frac{3\pi}{2},0)$, which are not critical points for the function $F$.
\begin{figure}
\centering
\includegraphics[scale =0.5]{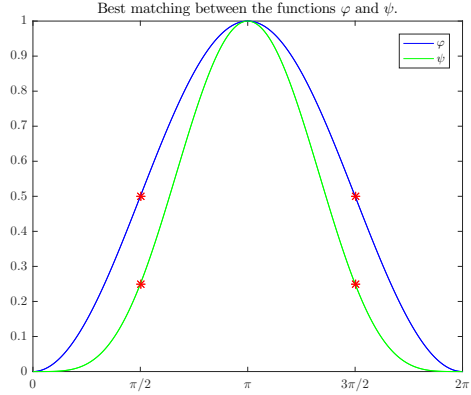}
\caption{Best matching between two Morse functions. The red dots represent the points at which the maximum pointwise distance between the two functions
is attained. Their abscissas are $\frac{\pi}{2}$ and $\frac{3\pi}{2}$.}
\label{fig1}
\end{figure}
\end{example}
The next proposition illustrates the possible position of the points that realize the natural pseudo-distance.
\begin{proposition}
\label{p2}
Let us assume that $\bar{d} \neq 0$ and take an optimal homeomorphism $\rho_{\bar{\alpha}}$. We define $\Theta = \{ \theta \in S^1  |  f_{\bar{\alpha}}(\theta):=F(\theta,\bar{\alpha})=\bar{d} \}$ and suppose that none of the points $(\theta,\bar{\alpha}), \theta \in \Theta$ is a critical point for $F$. Then:
\begin{equation}
\exists \bar{\theta}_1,\bar{\theta}_2 \in \Theta \text{ such that } \frac{\partial F}{\partial \alpha}(\bar{\theta}_1,\bar{\alpha})\frac{\partial F}{\partial \alpha}(\bar{\theta}_2,\bar{\alpha})<0.
\end{equation}
\end{proposition}
\begin{proof}
Since $\bar{d}\neq 0$, the function $F$ is at least $C^1$ in a neighborhood of the set $\Theta$.
We can consider the sequences defined in Lemma \ref{lem1}. We have that $(\theta_i^u,\alpha_i^u) \xrightarrow{i \rightarrow \infty} (\bar{\theta}_1,\bar{\alpha})$. 
Then since $F(\theta_i^u,\alpha_i^u) \geq F(\bar{\theta}_1,\bar{\alpha}) \geq F(\theta_i^u,\bar{\alpha})$,
\begin{equation}
\frac{\partial F}{\partial \alpha}(\bar{\theta}_1,\bar{\alpha}) =
\lim_{i\rightarrow \infty}\frac{F(\theta_i^u,\alpha_i^u) -F(\theta_i^u,\bar{\alpha})}{\alpha_i^u- \bar{\alpha}} \geq 0.
\end{equation}
This limit cannot be equal to $0$ because otherwise the point $(\bar{\theta}_1,\bar{\alpha})$ would be a critical point for the
function $F$, against our assumptions. Similarly, we have that  $F(\theta_i^d,\alpha_i^d) \geq F(\bar{\theta}_2,\bar{\alpha}) \geq F(\theta_i^d,\bar{\alpha})$,
\begin{gather*}
(\theta_i^d,\alpha_i^d) \xrightarrow{i \rightarrow \infty} (\bar{\theta}_2,\bar{\alpha}) \\
\frac{\partial F}{\partial \alpha}(\bar{\theta}_2,\bar{\alpha}) =
\lim_{i\rightarrow \infty}\frac{F(\theta_i^d,\alpha_i^d) -F(\theta_i^d,\bar{\alpha})}{\alpha_i^d- \bar{\alpha}} \leq 0.
\end{gather*}
Also in this case the limit cannot be equal to 0 because of our hypothesis, so that $\bar{\theta}_1 \neq \bar{\theta}_2$.
Therefore, $\frac{\partial F}{\partial \alpha}(\bar{\theta}_1,\bar{\alpha})\frac{\partial F}{\partial \alpha}(\bar{\theta}_2,\bar{\alpha})<0$.
\end{proof}
This proposition concerns cases similar to the one treated in the Example \ref{e2}.  
Let us consider the optimal homeomorphism $\rho_0 = id$, associated with $\bar{\alpha} = 0$. If we take $\rho_\varepsilon$, with an arbitrary small $\varepsilon>0$, then in a neighborhood of the point $\frac{\pi}{2}$ the values taken by the function $f_\varepsilon$ are smaller than the ones taken by the function $f_0$, but at the same time in a neighborhood of the point $\frac{3\pi}{2}$ the values taken by the function $f_\varepsilon$  are greater than the ones taken by the function $f_0$. We get an analogous result if we consider $\rho_{-\varepsilon}$, so that $\max_{\theta}f_\varepsilon(\theta) > \max_{\theta}f_0(\theta)=\bar{d}$ for any non-null $\varepsilon$ whose absolute value is small enough. 
\section{Localization of points that realize the natural pseudo-distance and finiteness of the set of optimal homeomorphisms}
\label{sec3}
In this section we will prove the main results of our paper. In Theorem \ref{th2} we summarize the conclusions obtained in the previous section and localize points at which, after a suitable alignment through an optimal homeomorphism, 
the pointwise distance between the two considered filtering functions takes a value equal to the natural pseudo-distance $\bar{d}$. This localization makes the computation of the natural pseudo-distance easier, as we will see in Example \ref{ex3}. Theorem \ref{th3} instead gives us conditions under which only a finite number of optimal homeomorphisms exists. 
Its usefulness does not reside only in the theoretical approach to the problem. From an applicative point of view it is convenient to know that there may be only a finite number of best matchings between two sets of data represented by real-valued functions. The proof of the theorem is based on the possibility of finding an infinite sequence of critical points for the function $F$, 
when the number of optimal homeomorphisms is infinite. This will lead to finding a degenerate point for the function $F$, and the study of the Hessian 
matrix at such a point will allow us to see that at least one of the functions taken into account must have a degenerate point, against the assumption that our functions are Morse.
\begin{theorem}
\label{th2}
Let $\varphi,\psi \in \mathcal{M}$ and $\bar{d}=d_{S^1}(\varphi,\psi)$. At least one of the following statements holds:
\begin{enumerate}
\item  There exist $\theta_1$ critical point for $\varphi$ and $\theta_2$ critical point for $\psi$ such that $\bar{d}= |\varphi(\theta_1)-\psi(\theta_2)|$; 

\item There exist  $\theta_1$, $\theta_2$, $\tilde{\theta}_1$ and $\tilde{\theta}_2$ such that $\bar{d}= |\varphi(\theta_1)-\psi(\theta_2)|=|\varphi(\tilde{\theta}_1)-\psi(\tilde{\theta}_2)|$ with
$$
\begin{cases}
\varphi'(\theta_1) =\psi'(\theta_2) \text{ and } \varphi'(\tilde{\theta}_1)=\psi'(\tilde{\theta}_2) \\
\theta_1 - \theta_2 = \tilde{\theta}_1 - \tilde{\theta}_2 \\
\varphi'(\theta_1)\varphi'(\tilde{\theta}_1) < 0 \quad (\text{or equivalently }\quad \psi'(\theta_2)\psi'(\tilde{\theta}_2) < 0)  
\end{cases} 
$$

$$
\text{if } \qquad (\varphi(\theta_1)-\psi(\theta_2))\cdot(\varphi(\tilde{\theta}_1)-\psi(\tilde{\theta}_2)) > 0 
$$
or
$$
\begin{cases}
\varphi'(\theta_1) =\psi'(\theta_2) \text{ and } \varphi'(\tilde{\theta}_1)=\psi'(\tilde{\theta}_2) \\
\theta_1 - \theta_2 = \tilde{\theta}_1 - \tilde{\theta}_2 \\
\varphi'(\theta_1)\varphi'(\tilde{\theta}_1) > 0 \quad (\text{or equivalently }\quad \psi'(\theta_2)\psi'(\tilde{\theta}_2) > 0)  
\end{cases}
$$
$$ 
\text{if } \qquad
 (\varphi(\theta_1)-\psi(\theta_2))\cdot(\varphi(\tilde{\theta}_1)-\psi(\tilde{\theta}_2)) < 0.
$$
\end{enumerate}
\end{theorem}

\begin{proof}
If $\bar{d} \neq 0$ the result follows from Theorem \ref{p1} and Proposition \ref{p2}, by setting $\bar{\alpha}=\theta_1-\theta_2$. If $\bar{d} = 0$ there must exist a $\rho \in S^1$ such that $\varphi=\psi \circ \rho$ and therefore both the statements of the theorem are trivially true.
\end{proof}

Our last result concerns the finiteness of the set of optimal homeomorphisms. 

\begin{theorem}
\label{th3}
The number of optimal homeomorphisms between two functions $\varphi,\psi \in \mathcal{M}$ is finite.
\end{theorem}

\begin{proof}
We will see that if there is an infinite number of optimal homeomorphisms then at least one of the two functions $\varphi,\psi$ has a degenerate point.\\
Suppose there is an infinite family of optimal homeomorphisms between $\varphi$ and $\psi$. For the compactness of the Lie group $S^1$ we can 
construct a sequence $(\rho_{\alpha_i})$ of optimal homeomorphisms different from each other, which converges to an optimal homeomorphism $\rho_{\bar{\alpha}}$. By possibly extracting a subsequence, we can assume either $\alpha_i > \bar{\alpha}$ for all $i$ 
or $\alpha_i < \bar{\alpha}$ for all $i$. We will confine ourselves to examine just the first case, since the other can be managed in a
completely analogous way. Once again for the compactness of $S^1$, we can find a sequence $(\theta_i)$ that converges to a point $\bar{\theta}$, such that
\begin{equation*}
f_{\alpha_i}(\theta_i)=\max_{\theta} f_{\alpha_i}(\theta)=\bar{d} \quad \text { and } \quad
f_{\bar{\alpha}}(\bar{\theta})=\max_{\theta} f_{\bar{\alpha}}(\theta)=\bar{d}.
\end{equation*}
The sequence $((\theta_i,\alpha_i))$ converges to $(\bar{\theta},\bar{\alpha})$.
We see that the following inequalities hold for all $i$:
\begin{gather}
F(\theta_i,\alpha_i)-F(\theta_i,\bar{\alpha}) \geq F(\theta_i,\alpha_i)-F(\bar{\theta},\bar{\alpha}) =\bar{d}-\bar{d}= 0 \\
F(\bar{\theta},\alpha_i)-F(\bar{\theta},\bar{\alpha}) \leq  F(\theta_i,\alpha_i)-F(\bar{\theta},\bar{\alpha}) = \bar{d}-\bar{d}=0.
\end{gather}
Now we will consider the case in which $\bar{d}>0$ and therefore the function $F$ is a $C^1$ function near the points that realize the natural pseudo-distance.
We have that
\begin{gather}\label{dis1}
\frac{\partial F}{\partial \alpha}(\bar{\theta},\bar{\alpha}) =
\lim_{i \rightarrow \infty} \frac{F(\theta_i,\alpha_i)-F(\theta_i,\bar{\alpha}) }{\alpha_i - \bar{\alpha}} \geq 0 \\ \label{dis2}
\frac{\partial F}{\partial \alpha}(\bar{\theta},\bar{\alpha}) =
\lim_{i \rightarrow \infty} \frac{F(\bar{\theta},\alpha_i)-F(\bar{\theta},\bar{\alpha}) }{\alpha_i - \bar{\alpha}} \leq 0 
\end{gather}
and hence $(\bar{\theta},\bar{\alpha})$ is a critical point for $F$.
Now we want to show there is an infinite sequence of critical points for $F$, proving that $F$ is not a Morse function and that a degenerate critical point exists.  \\
Let us consider the continuous function $g(\alpha)=\max_{\theta\in S^1}F(\theta,\alpha)$. We know that $g(\alpha_i)=\bar{d}$ 
for every index $i$, and that the function $g$  must attain its maximum value over every compact interval $[\alpha_i,\alpha_{i+1}]$. Let us call $\tilde{\alpha}_i$ the point at which the restriction of $g$ to the set  $[\alpha_i,\alpha_{i+1}]$ takes its maximum value. We can now consider the point $\tilde{\theta}_i$ of global maximum for the function $f_{\tilde{\alpha}_i}$. The pair $(\tilde{\theta}_i , \tilde{\alpha}_i )$ is a maximum point for the function $F$, and hence also a critical point, since $F$ is at least $C^1$ in a neighbourhood of that pair.

Since there exists a  point $(\bar{\theta},\bar{\alpha})$ that is an accumulation point of critical points for the function $F$, 
$(\bar{\theta},\bar{\alpha})$ is also a degenerate point for the function $F$. This means that the Hessian matrix at $(\bar{\theta},\bar{\alpha})$ has determinant equal to $0$. The Hessian matrix is

\begin{equation} 
\textrm{sgn}(\varphi(\bar{\theta})-\psi(\bar{\theta}+\bar{\alpha}))\cdot
\left(
\begin{array}{c c}
\varphi''(\bar{\theta})-\psi''(\bar{\theta}+\bar{\alpha}) & - \psi''(\bar{\theta}+\bar{\alpha})  \\
-\psi''(\bar{\theta}+\bar{\alpha})  & -\psi''(\bar{\theta}+\bar{\alpha})  \\
\end{array} \right) 
\end{equation}
(where we use the notation: $\varphi''(x)= \frac{d^2 \varphi}{dx^2}(x)$ and the same for $\psi$).
If $\psi''(\bar{\theta}+\bar{\alpha}) \neq 0$,  it must be $\varphi''(\bar{\theta})-\psi''(\bar{\theta}+\bar{\alpha}) = -\psi''(\bar{\theta}+\bar{\alpha})$, and hence $\varphi''(\bar{\theta})=0$.
Since $(\bar{\theta},\bar{\alpha})$ is a critical point for $F$ it must also be:
\begin{equation}
\frac{\partial F}{\partial \theta}(\bar{\theta},\bar{\alpha})= 0 \qquad \text{and} \qquad \frac{\partial F}{\partial \alpha}(\bar{\theta},\bar{\alpha})= 0.
\end{equation} 
The equality $\frac{\partial F}{\partial \alpha}(\bar{\theta},\bar{\alpha})= 0$ implies that $-\psi'(\bar{\theta}+\bar{\alpha})=0$. Since$\frac{\partial F}{\partial \theta}(\bar{\theta},\bar{\alpha})= 0$ it follows that $\varphi'(\bar{\theta})-\psi'(\bar{\theta},\bar{\alpha})=0$, and it must be $\varphi'(\bar{\theta})=0$. Then $\bar{\theta}$ is a degenerate point for the function $\varphi$, which cannot be a Morse function.\\
Analogously, in the case $\psi''(\bar{\theta}+\bar{\alpha})=0$, we have that $\psi'(\bar{\theta}+\bar{\alpha})=0$ and hence $\bar{\theta}+\bar{\alpha}$ is a degenerate point for the function $\psi$, which cannot be a Morse function. \\
We eventually consider the case in which $\bar{d}=0$. If $\varphi$, $\psi$ are Morse functions and $\bar{d}=0$, there cannot exist an infinite number of homeomorphisms $\rho_\alpha$ such that $\varphi = \psi \circ \rho_\alpha$, since these homeomorphisms must match the critical points of the two functions, and the number of these points is finite. 
\end{proof}

\begin{example}
\label{ex3}
The use of Theorems \ref{th2} and \ref{th3} can be clarified by the following example, illustrating the computation of the natural pseudo-distance between the two functions $\varphi(\theta)=\frac{1}{2}\sin(2\theta)$ and $\psi(\theta)=\sin(\theta)$. Since $\varphi$ and $\psi$ are Morse, Theorem \ref{th3} guarantees that the number of optimal homeomorphisms between $\varphi$ and $\psi$ is finite. We can apply Theorem \ref{th2} to find these homeomorphisms. Since the difference between the maximum values of the two functions is $\frac{1}{2}$, the inequality $\bar{d}\geq \frac{1}{2}$ holds. Let us look for the points described in the second statement of Theorem \ref{th2}. For any $\alpha \in S^1$ we want to find the corresponding $\theta$s that satisfy the equation $\varphi'(\theta)-\psi'(\theta+\alpha)$, i.e.
\begin{equation}
\cos(2\theta)=\cos(\theta+\alpha).
\end{equation}
We obtain the following solutions depending on $\alpha$:
\begin{equation*}
\bar{\theta}_1 = \alpha \quad \bar{\theta}_2 = -\frac{\alpha}{3} \quad \bar{\theta}_3 = -\frac{\alpha}{3} +\frac{2\pi}{3} \quad \bar{\theta}_4 = -\frac{\alpha}{3} +\frac{4\pi}{3}.
\end{equation*}
We insert these values in the equation $|\varphi(\theta_1)-\psi(\theta_2)| = |\varphi(\tilde{\theta}_1)-\psi(\tilde{\theta}_2)|$ (see Theorem \ref{th2}) to find the possible $\alpha$ associated with the optimal homeomorphisms. If $\rho_\alpha$ is an optimal homeomorphism, two indexes $i$ and $j$ must exist such that
\begin{equation*}
\bar{d}=
\left | \frac{1}{2}\sin(2\bar{\theta}_i)-\sin(\bar{\theta}_i+\alpha)\right | =\left | \frac{1}{2}\sin(2\bar{\theta}_j)-\sin(\bar{\theta}_j+\alpha)\right |.
\end{equation*} 
This leads to consider these possible values for $\alpha$: $0, \frac{\pi}{4}, \frac{\pi}{2}, \frac{3 \pi}{4}, \pi, \frac{5 \pi}{4}, \frac{3 \pi}{2}, \frac{7\pi}{4} $. With some calculations we see that 
\begin{gather}
\max_{\theta \in S^1} f_\alpha(\theta) = \frac{3 \sqrt{3}}{4} \qquad \text{ for } \alpha \in \left\{0, \frac{\pi}{2}, \pi, \frac{3\pi}{2} \right\} \\
\max_{\theta \in S^1} f_\alpha(\theta) = \frac{3}{2} \qquad \text{ for } \alpha \in \left\{ \frac{\pi}{4},  \frac{3 \pi}{4},  \frac{5 \pi}{4}, \frac{7\pi}{4}\right\}.
\end{gather}
Hence the natural pseudo-distance is $\frac{3 \sqrt{3}}{4}$ and the points $\theta_1=\theta_2=\frac{2\pi}{3}$ and $\tilde{\theta}_1=\tilde{\theta}_2=\frac{4\pi}{3}$ satisfy the conditions of the second statement of Theorem \ref{th2}.
\end{example}
\section{Acknowledgments}
The author thanks Patrizio Frosini for his helpful advice.

\providecommand{\bysame}{\leavevmode\hbox to3em{\hrulefill}\thinspace}
\providecommand{\MR}{\relax\ifhmode\unskip\space\fi MR }
\providecommand{\MRhref}[2]{%
  \href{http://www.ams.org/mathscinet-getitem?mr=#1}{#2}
}
\providecommand{\href}[2]{#2}
\bibliographystyle{plain}
\bibliography{art}
\end{document}